\documentclass
[twocolumn, superscriptaddress,
showpacs,preprintnumbers,amsmath,amssymb]{revtex4}

\usepackage{graphicx}
\usepackage{amssymb,amsmath,amsthm}
\usepackage{epsfig}
\newtheorem{Thm}{Theorem}
\newtheorem{Cor}{Corollary}
\newtheorem{Lem}[Thm]{Lemma}
\newtheorem{Prop}{Proposition}

\theoremstyle{definition}
\newtheorem{Def}{Definition}

\newcommand{\bra}[1]{{\left\langle #1 \right|}}
\newcommand{\ket}[1]{{\left| #1 \right\rangle}}

\newcommand{\R}{\mbox{$\mathbb R$}}

\newcommand{\T}{\mbox{$\mathrm{tr}$}}

\begin{document}
\title{Generalized Entropy and Global Quantum Discord in Multi-party Quantum systems}

\author{Dong Pyo Chi}
\affiliation{
 Department of Mathematical Sciences, Seoul National University, Seoul 151-742, Korea
}

\author{Jeong San Kim}
\email{freddie1@suwon.ac.kr} \affiliation{
 Department of Mathematics, University of Suwon, Kyungki-do 445-743, Korea
}

\author{Kyungjin Lee}
\affiliation{
 Department of Mathematical Sciences, Seoul National University, Seoul 151-742, Korea
}
\date{\today}

\begin{abstract}
Using Tsallis-$q$ entropy, we introduce the generalized concept of global quantum discord, namely the $q$-global quantum discord,
and provide its analytic evaluation for two classes of multi-qubit states.
We also provide a sufficient condition, for which the pairwise quantum correlations in terms of $q$-global quantum discord is monogamous in multi-party quantum systems.
\end{abstract}

\maketitle

\section{Introduction}

Traditionally, quantum entanglement was believed as the unique quantum correlation
that enables us to overcome the limit of classical computational models~\cite{tele, qkd1, qkd2}.
However, it has been recently shown that entanglement is not the only quantum correlation that can be used to obtain a
quantum speed-up; there exist quantum computational models such as
the deterministic quantum computation with one qubit (DCQ1), which only uses separable states~\cite{KL,DFC,DV}.
The resource believed to provide an enhancement in this computational task is {\em quantum discord} (QD)~\cite{OZ,HV}.

Besides quantum computational protocols, many other remarkable applications
of QD have also been proposed such as the characterization of quantum phase transitions~\cite{Sar} and the dynamics of quantum systems
under decoherence~\cite{Maz}. QD is thus identified as a general resource in quantum information processing.

Since the original definition of QD in bipartite quantum systems~\cite{OZ},
which considers a set of local measurement only on one subsystem,
a symmetric extension of QD was suggested, namely {\em global quantum discord} (GQD)~\cite{RS}, with analytical expressions
for some classes of quantum states~\cite{Xu}. Whereas QD is defined in terms of von Neumann entropy of quantum states,
a generalized version of QD in terms of Tsallis entropy has
also been proposed for bipartite quantum systems~\cite{MPP,J}.

For efficient applications of QD as a resource in quantum information and communication protocols,
it is an important task to characterize its possible distribution or shareability in multi-party quantum systems.
For examples, quantum entanglement cannot be shared freely in multi-party quantum systems, and
this restricted shareability of entanglement is known as the {\em monogamy of entanglement} (MoE)~\cite{T04}.

Mathematically, MoE in multi-party quantum systems is characterized as a trad-off inequality in terms of
bipartite entanglement measures.
The first monogamy inequality of entanglement was established in three-qubit systems using tangle as the bipartite
entanglement measure~\cite{CKW}. Since then, there have been intensive
research contributed on this topic for possible generalization of the monogamy inequality into multi-party higher-dimensional quantum systems~\cite{OV}.
Recently, it has been shown that the entanglement monogamy inequality holds for arbitrary dimensional multi-party
quantum systems in terms of the squashed entanglement~\cite{BCY10}.

MoE plays a crucial role in many quantum information processing tasks.
In quantum key-distribution protocols, for example, the possible amount of information an eavesdropper could obtain about
the secret key can be restricted by MoE, which is the fundamental concept of security proof~\cite{Paw}.
Thus the founding principle of quantum cryptographic schemes that an eavesdropper cannot obtain any information
without disturbance is guaranteed by MoE, the law of quantum physics, rather than assumptions on the difficulty of computation.

Because MoE is a property of a typical quantum correlation, quantum entanglement, without any classical counterpart,
it is also natural and meaningful to investigate other quantum correlations such as QD or GQD can have such restricted
shareability or distribution in multi-party quantum systems, which still has many important open questions.~\cite{CKW,OV,Z,K}.

In this paper, using Tsallis-$q$ entropy, we propose a one-parameter class of quantum correlation measures,
{\em $q$-global quantum discord} ($q$-GQD), which include GQD as a special case.
We provide an analytic expression of $q$-GQD for some classes of multi-qubit states,
and show that the pairwise quantum correlations in terms of global quantum-$q$ discord can be monogamous
in multi-party quantum systems.

This paper is organized as follows. In Section~\ref{Sec: QD}, we briefly recall the motivation and definitions of QD and GQD.
In Section~\ref{Sec:q-GQD}, we introduce the concept of $q$-GQD in multi-party quantum systems
as well as its properties.
In Section~\ref{Sec:analytic}, we provide an analytical expression of $q$-GQD for some classes of multi-qubit states, and
we show a sufficient condition for monogamy inequality of $q$-GQD in multi-party quantum systems in Section~\ref{Sec:mono}.
In Section~\ref{Sec:Con}, we summarize our results.

\section{Quantum Discord}\label{Sec: QD}
\subsection{Quantum Discord in Bipartite and Multipartite Quantum Systems}
\label{subsec:QD}
For a bipartite quantum state $\rho^{AB}$
and a set of local von Neumann measurement $\{\Pi_{j}^{B}\}$ on subsystem $B$,
the quantum state of subsystem $A$ after performing the measurement  $\{\Pi_{j}^{B}\}$ and obtaining
the measurement outcome $j$ is $\rho^{A}_{j}= \T_{B} \left[\left(I_{A}\otimes\Pi_{j}^{B}\right)\rho^{AB}\right]/p_j$,
with probability $p_{j}=\T\left[\left(I_{A}\otimes\Pi_{j}^{B}\right)\rho^{AB}\right]$.
The average of the von Neumann entropies $S(\rho^{A}_{j})$ of $\rho^{A}_{j}$ weighted by probabilities $p_{j}$ yields to
the conditional entropy of subsystem $A$ given the complete measurement $\{\Pi_{j}^{B}\}$ on subsystem $B$,
\begin{align}
S^{\{\Pi_{j}^{B}\}}\left(\rho^A\right)=\sum _{j} p_j S\left(\rho^{A}_{j}\right),
\label{condent}
\end{align}
where $S(\rho)=-\T\rho\log\rho$ is the von Neumann entropy of the quantum state $\rho$.

We note that the possible inequivalent concepts of quantum mutual information we can consider here are
\begin{align}
{\mathcal I}(\rho^{AB})= S\left(\rho^A\right)+S\left(\rho^B\right)-S\left(\rho^{AB}\right)
\label{qmut1}
\end{align}
and
\begin{align}
{\mathcal J}^{\{\Pi_{j}^{B}\}}\left(\rho^{AB}\right) = S\left(\rho^A\right)-S_{\{\Pi_{j}^{B}\}}\left(\rho^A\right)
\label{qmut2}
\end{align}
where $\rho^A=\T_B\rho^{AB}$, $\rho^B=\T_A\rho^{AB}$ are the reduced density matrices of $\rho_{AB}$ on
subsystems $A$ and $B$, respectively. In other words, Eq.~(\ref{qmut1}) is a straightforward generalization of
the classical mutual information in terms of
the quantum conditional entropy $S(\rho^{A|B})=S(\rho^{AB})-S(\rho^{B})$,
whereas Eq.~(\ref{qmut2}) can be considered as the measurement-induced quantum mutual information
using Eq.~(\ref{condent}).

The QD of $\rho^{AB}$ is then defined by the minimized
difference between these two inequivalent generalizations of classical mutual information over all possible
von Neumann measurements $\{\Pi_{j}^{B}\}$ on subsystem $B$~\cite{OZ},
\begin{align}
{\delta}^{\leftarrow}\left(\rho^{AB}\right)=\min_{\{\Pi_{i}^{B}\}}\left[{\mathcal I}\left(\rho^{AB}\right)
-{\mathcal J}^{\{\Pi _{i}^{B}\}}\left(\rho^{AB}\right)\right].
\label{QD}
\end{align}
QD is not symmetric under the interchange of subsystems to be measured,
\begin{align}
{\delta}^{\leftarrow}\left(\rho^{AB}\right)\neq {\delta}^{\leftarrow}\left(\rho^{BA}\right),
\label{notsym}
\end{align}
and it is nonnegative for any bipartite quantum state $\rho^{AB}$~\cite{OZ}.

QD has been generalized into multi-party quantum systems~\cite{OW};
for an $n$-party quantum states $\rho^{A_1\cdots A_n}=\rho^{\bf A}$ with the reduced density matrix $\rho^{A_k}$
of the subsystem $A_k$ for each $k=1,\cdots, n$, its quantum mutual information is given by
\begin{equation}
{\mathcal I}\left(\rho^{{\bf A}}\right)=\sum_{i=1}^{n} S\left(\rho^{A_{i }}\right)-S\left(\rho^{{\bf A}}\right).
\label{nqmu}
\end{equation}
(Throughout this paper, the bold superscript ${\bf A}$ in $\rho^{\bf A}$ denotes an $n$-party quantum system unless specified.)
By using the notation of quantum conditional entropy
\begin{align}
S\left(\rho^{\overline{A}_k|A_{k}}\right)=S\left(\rho^{{\bf A}}\right)-S\left(\rho^{A_{k}}\right),
\label{qcon}
\end{align}
where the superscript $\overline{A}_k$ stands for the quantum systems complement to $A_k$,
we can rewrite the quantum mutual information in Eq.~(\ref{nqmu}) as
\begin{equation}
{\mathcal I}\left(\rho^{{\bf A}}\right)=\sum_{i \ne k} S\left(\rho^{A_{i}}\right)-S\left(\rho^{{\overline{A}_k|A_{k}}}\right).
\label{nqmu2}
\end{equation}

Now let us consider the situation that a von Neumann measurement $\{\Pi _{j}^{A_{k}} \}$ is performed on subsystem $A_k$.
The post-measurement joint state of the systems $\overline{A}_k$ is given by
\begin{equation}
 \rho^{\overline{A}_k}_j=\T _{A_{k}} \left[\left(P_{j}^{A_{k}}\right) \rho^{\bf A}\right]/p_{j}^{A_{k}}
\end{equation}
where $P_{j}^{A_{k}} = (I \otimes \cdots \otimes \Pi _{j}^{A_{k}} \otimes \cdots \otimes I)$ is the measurement operator acting only on
the subsystem $A_k$ and $p_{j}^{A_{k}}=\T[P_{j}^{A_{k}} \rho^{{\bf A}}]$
is the probability of the measurement result on subsystem $A_k$.
Similar to Eq.~(\ref{condent}), the average of the von Neumann entropies $S(\rho^{\overline{A}_k}_j)$ weighted
by the probabilities $p_{j}^{A_{k}}$
leads us to the quantum conditional entropy
given the complete measurement $\{\Pi _{j}^{A_{k}} \}$ on the system $A_{k}$
\begin{equation}
S^{\{\Pi _{j}^{A_{k}} \}}\left(\rho^{\overline{A}_k}\right)=\sum_{j} p_{j}^{A_{k}}S\left(\rho^{\overline{A}_k}_j\right).
\label{ncondent}
\end{equation}
Thus the quantum mutual information, induced by the von Neumann measurement $\{\Pi _{j}^{A_{k}} \}$ is defined by
\begin{equation}
{\mathcal J}^{\{\Pi _{j}^{A_{k}}\}}\left(\rho^{{\bf A}}\right)=\sum_{i \ne k}S\left(\rho^{A_{i}}\right)-S^{\{\Pi _{j}^{A_{k}} \}}\left(\rho^{\overline{A}_k}\right),
\label{nqmu3}
\end{equation}
which is an analogous quantity of Eq.~(\ref{qmut2}).
The QD for the $n$-party state $\rho^{{\bf A}}$ is defined as
the minimized difference between Eq~(\ref{nqmu2}) and Eq.~(\ref{nqmu3}) over all possible
von Neumann measurement $\{\Pi _{j}^{A_{k}} \}$ on subsystem $A_k$,
\begin{eqnarray}
{\delta}^{A_{k}}(\rho^{{\bf A}})=\min_{\{\Pi _{j}^{A_{k}}\}}\left[{\mathcal I}\left(\rho^{{\bf A}}\right)-
{\mathcal  J}^{\{\Pi _{j}^{A_{k}} \}}\left(\rho^{{\bf A}}\right)\right].
\end{eqnarray}

\subsection{Global Quantum Discord}
\label{subsec:GQD}
For a bipartite quantum state $\rho^{AB}$,
its mutual information
${\mathcal I}(\rho^{AB})$ in Eq.~(\ref{qmut1}) can be expressed in terms of the relative entropy between $\rho^{AB}$ and
$\rho^{A}\otimes \rho^{B}$,
\begin{equation}
{\mathcal I}\left(\rho^{AB}\right)=S\left(\rho^{AB}||\rho^{A}\otimes \rho^{B}\right),
\label{qmurel}
\end{equation}
where $S\left(\rho||\sigma\right)=\T\rho\log \rho-\T\rho\log\sigma$ is the quantum relative entropy of $\rho$ and $\sigma$.

In order to express the measurement-induced quantum mutual information in Eq.~(\ref{qmut2}) in terms of quantum relative entropy,
we consider a non-selective von-Neumann measurement $\Phi=\{\Pi_{j}^{B}=\ket{b_j}\bra{b_j}\}$ on subsystem $B$,
which yields to the quantum state
\begin{align}
\Phi\left(\rho^{AB}\right)=&\sum_{j}\left(I \otimes \Pi _{j}^{B}\right)\rho^{AB}\left(I \otimes \Pi _{j}^{B}\right)\nonumber\\
=&\sum_{j} p_{j} \rho^{A}_j \otimes \ket{b_{j}}\bra{b_{j}},
\label{Phirhoab}
\end{align}
and its reduced density matrix of subsystem $B$
\begin{equation}
\Phi\left(\rho^{B}\right)=\T_{A}\left(\Phi\left(\rho^{AB}\right)\right)=\sum_{j} p_{j} \ket{b_{j}} \bra{b_{j}}.
\label{Phirhob}
\end{equation}

Because
$\{\ket{b_j}\}$ forms an orthonormal basis for subsystem $B$,
we have
\begin{equation}
S\left(\Phi\left(\rho^{AB}\right)\right)=H(p)+\sum_{j} p_{j} S(\rho^{A}_j)
\end{equation}
and
\begin{equation}
S\left(\Phi\left(\rho^{B}\right)\right)=H(P),
\end{equation}
where $H(P)$ is the Shannon entropy of the probability ensemble $P=\{p_i\}$.
Thus we can rewrite the measurement-induced quantum mutual information in Eq.~(\ref{qmut2}) as
\begin{align}
{\mathcal J}^{\{\Pi_{j}^{B}\}}\left(\rho^{AB}\right)&=S\left(\rho^{A}\right)-\sum_{j} p_{j} S\left(\rho^{A}_j\right)\nonumber\\
=&S\left(\rho^{A}\right)+S\left(\Phi\left(\rho^{B}\right)\right)-S\left(\Phi\left(\rho^{AB}\right)\right)\nonumber\\
=&S\left(\Phi(\rho^{AB})||\rho^{A}\otimes\Phi\left(\rho^{B}\right)\right)\nonumber\\
=&{\mathcal I}\left(\Phi\left(\rho^{AB}\right)\right).
\label{qmucon}
\end{align}
From Eqs.~(\ref{qmurel}) and (\ref{qmucon}), the definition of QD in Eq.~(\ref{QD}) can be expressed in terms of the relative entropy as
\begin{align}
{\delta}^{\leftarrow}\left(\rho^{AB}\right)=&\min_{\{\Pi_{j}^{B}\}}
[S\left(\rho^{AB}||\rho^{A}\otimes \rho^{B}\right)\nonumber\\
&~~~~~~~~~-S\left(\Phi\left(\rho^{AB}\right)||\rho^{A}\otimes\Phi\left(\rho^{B}\right)\right)]\nonumber\\
=& \min_{\{\Pi_{j}^{B}\}}
\left[{\mathcal I}\left(\rho^{AB}\right)-{\mathcal I}\left(\Phi\left(\rho^{AB}\right)\right)\right],
\label{relQD}
\end{align}
where the minimum is taken over all possible measurement $\{\Pi_{j}^{B}\}$ on subsystem $B$.

The GQD of a bipartite state $\rho^{AB}$
was defined by considering von Neumann measurements $\{\Pi_{j_1}^{A} \otimes \Pi_{j_2}^{B} \}$ on both subsystems $A$ and $B$,
\begin{align}
{\mathcal D}\left(\rho^{AB}\right)
=&\min_{\{\Pi _{j_1}^{A}\otimes \Pi _{j_2}^{B} \}}\left[{\mathcal I}\left(\rho^{AB}\right)-{\mathcal I}\left(\Phi\left(\rho^{AB}\right)\right)\right],
\label{GQD}
\end{align}
where
\begin{equation}
\Phi\left(\rho^{AB}\right)=\sum_{j_1,j_2}\left(\Pi _{j_1}^{A} \otimes \Pi _{j_2}^{B}\right)\rho^{AB}\left(\Pi _{j_1}^{A} \otimes \Pi _{j_2}^{B}\right),
\end{equation}
and the minimization is over all possible von Neumann measurements $\{\Pi_{j_1}^{A}\}$ and $\{\Pi_{j_2}^{B}\}$ on
subsystems $A$ and $B$, respectively.

Unlike QD, it is clear from the definition that GQD is symmetric under the permutation of subsystems. GQD was also shown to be
nonnegative for an arbitrary quantum state.
Moreover GQD has an useful operational interpretation; in the absence of GQD, the quantum state simply describes
a classical probability distribution.

The definition of GQD in Eq.~(\ref{GQD}) was also generalized to multi-party quantum systems;
for an $n$-party quantum state $\rho^{{\bf A}}$, its GQD is defined as
\begin{align}
{\mathcal D}\left(\rho^{\bf A}\right)
=& \min_{\{\Pi _{j}\}}[{\mathcal  I}(\rho^{{\bf A}})- {\mathcal  I}\left(\Phi\left(\rho^{{\bf A}}\right) \right)],
\label{nGQD}
\end{align}
where $\Phi(\rho^{{\bf A}}) = \sum_{j} \Pi_{j}  \rho^{{\bf A}} \Pi_{j} $ is the density operator after a non-selective local measurement
$\Phi=\{\Pi_{j}= \Pi_{j_{1}}^{A_{1}} \otimes\Pi_{j_{2}}^{A_{2}}\otimes\cdots\otimes\Pi_{j_{n}}^{A_{n}}\}$ with $j$ denoting the index string $(j_{1},~\cdots ,j_{n})$,
\begin{align}
{\mathcal  I}(\Phi(\rho^{{\bf A}}))=\sum_{k=1}^{n} S\left(\Phi\left(\rho^{A_{k }}\right)\right)-S\left(\Phi\left(\rho^{{\bf A}}\right)\right)
\label{npimu}
\end{align}
and
\begin{align}
\Phi\left(\rho^{A_{k}}\right)=&\T_{\overline{A}_k}\Phi\left(\rho^{{\bf A}}\right)
\label{npii}
\end{align}
is the reduced density operator of $\Phi(\rho^{{\bf A}})$ onto the subsystem $A_k$ for each $k=1,\cdots, n$.

\section{Global Quantum Discord in terms of Tsallis-$q$ Entropy}
\label{Sec:q-GQD}
In this section, we first recall the definition of Tsallis-$q$ entropy of quantum states~\cite{tsallis, lv} and
bipartite QD concerned with Tsallis-$q$ entropy~\cite{J}. We then introduce a one-parameter class
of GQD in terms of Tsallis-$q$ entropy and investigate its properties.

\subsection{Tsallis-$q$ Entropy and Quantum Discord}

Using the generalized logarithmic function with respect to the parameter $q$ (namely $q$-logarithm)
\begin{eqnarray}
\ln _{q} x &=&  \frac {x^{1-q}-1} {1-q},
\label{qlog}
\end{eqnarray}
quantum Tsallis-$q$ entropy of a quantum state $\rho$ is defined as
\begin{align}
S_{q}\left(\rho\right)=-\T \rho ^{q} \ln_{q} \rho = \frac {1-\T\left(\rho ^q\right)}{q-1}
\label{Tqent}
\end{align}
for $q > 0,~q \ne 1$~\cite{lv}.
Although the quantum Tsallis-$q$ entropy has a singularity at $q=1$,
it is straightforward to check that it converges to von Neumann entropy when $q$ tends to $1$,
\begin{equation}
\lim_{q\rightarrow 1}S_{q}\left(\rho\right)=S\left(\rho\right).
\end{equation}

Tsallis entropy has been widely used in many areas of quantum
information theory such as the conditions for separability of quantum
states~\cite{ar,tlb,rc} and the characterization of classical statistical correlations
inherented in quantum states~\cite{rr}. There are also
discussions about using the non-extensive statistical
mechanics to describe quantum entanglement~\cite{bpcp}.
Similar to other entropy functions, Tsallis entropy is nonnegative for any quantum state.

Using Tsallis-$q$ entropy, QD was generalized to a one-parameter class of quantum correlation measure, namely
$q$-quantum discord ($q$-QD) ~\cite{MPP,J}. The inequivalent expressions of quantum mutual information
in Eqs.~(\ref{qmut1}) and (\ref{qmut2}) for a bipartite quantum state $\rho^{AB}$
can be generalized in terms of Tsallis-$q$ entropy as
\begin{align}
{\mathcal I}_q\left(\rho^{AB}\right)= S_{q}\left(\rho^A\right)+S_{q}\left(\rho^B\right)-S_{q}\left(\rho^{AB}\right)
\label{Tqmut1}
\end{align}
and
\begin{align}
{\mathcal J}_{q}^{\{\Pi_{j}^{B}\}}\left(\rho^{AB}\right) = S_{q}\left(A\right)-S_{q}^{\{\Pi_{j}^{B}\}}\left(\rho^A\right),
\label{Tqmut2}
\end{align}
where
\begin{align}
S_{q}^{\{\Pi_{j}^{B}\}}\left(\rho^A\right)=\sum _{j} p^q_j S_q\left(\rho^{A}_{j}\right)
\label{qcondent}
\end{align}
is the $q$-expected value of Tsallis-$q$ entropies $S_q\left(\rho^{A}_{j}\right)$~\cite{J}.
Then the $q$-QD for a bipartite state $\rho^{AB}$ is defined as
\begin{align}
\delta_{q}^{\leftarrow}\left(\rho^{AB}\right)=\min_{ \{ \Pi _{j}^{B}\}}\left[{\mathcal I}_{q}(\rho^{AB})-{\mathcal J}_{q}^{\{\Pi_{j}^{B}\}}\left(\rho^{AB}\right)\right],
\label{qQD}
\end{align}
where the minimization is taken over all possible sets of rank-one measurement $\{\Pi_{j}^{B}\}$ on subsystem $B$.

Due to the continuity of Tsallis-$q$ entropy with respect to the parameter $q$, $q$-QD converges to QD when $q$ tends to 1, therefore
the nonnegativity of $q$-QD follows from that of QD as $q$ tends to 1. Furthermore, the following proposition provides
a possible range of $q$ where $q$-QD is nonnegative~\cite{J, MPP}.
\begin{Prop}
For any bipartite state $\rho^{AB}$ and $0 < q \leq1$,
\begin{align}
\delta_{q}^{\leftarrow}\left(\rho^{AB}\right)\geq 0.
\label{q-QDposi}
\end{align}
\label{Prop:q-QDposi}
\end{Prop}

\subsection{$q$-Global Quantum Discord}

For a non-selective von-Neumann measurement $\Phi=\{\Pi_{j}^{B}=\ket{b_j}_B\bra{b_j}\}$ on subsystem $B$ of $\rho^{AB}$,
which yields a bipartite quantum state $\Phi\left(\rho^{AB}\right)$ in
Eq.~(\ref{Phirhoab}) and the reduced density matrix $\Phi\left(\rho^{B}\right)$ in Eq.~(\ref{Phirhob}), we have
\begin{align}
S_{q}\left(\Phi\left(\rho^{AB}\right)\right)=& \frac{1-\T\left(\sum_{j} p_{j} \rho_j^{A}\otimes\ket{b_{j}}_B\bra{b_{j}}\right)^{q}}{q-1}\nonumber\\
  =&\frac {1-\T\sum_{j}p_{j}^{q} \left(\rho_j^{A}\right)^{q}}{q-1}\nonumber\\
  =&\frac {1-\sum_{j}p_{j}^{q}}{q-1} +\sum_{j}p_{j}^{q}\frac{1-\T\left(\rho_j^{A}\right)^{q}}{q-1}\nonumber\\
  =&S_{q}\left(\Phi\left(\rho^{B}\right)\right)+\sum_{j}p_{j}^{q}S_{q}\left(\rho_j^{A}\right),
\label{qcondent2}
\end{align}
where the last equality is from the definition of Tsallis-$q$ entropy for
quantum states $\Phi\left(\rho^{B}\right)$ and $\rho_j^{A}$, respectively.

From Eq.~(\ref{Tqmut2}) together with Eqs.~(\ref{qcondent}) and (\ref{qcondent2}), we have
\begin{align}
{\mathcal J}_{q}^{\{\Pi_{j}^{B}\}}\left(\rho^{AB}\right)
=&S_{q}\left(\rho^{A}\right)+S_{q}\left(\Phi\left(\rho^{B}\right)\right)-S_{q}\left(\Phi\left(\rho^{AB}\right)\right)\nonumber \\
=&{\mathcal I}_{q}\left(\Phi\left(\rho^{AB}\right)\right),
\label{TrelQD}
\end{align}
which leads us to the following lemma.
\begin{Lem}
For a bipartite state $\rho^{AB}$, its $q$-quantum discord can be expressed as
\begin{align}
\delta_{q}^{\leftarrow}\left(\rho^{AB}\right)=\min_{\Phi}\left[{\mathcal I}_{q}\left(\rho^{AB}\right)- {\mathcal I}_{q}\left(\Phi\left(\rho^{AB}\right)\right)\right],
\label{q-QD2}
\end{align}
where the minimization is taken over all possible von Neumann measurement $\Phi$ of subsystem $B$.
\label{Def:q-QD2}
\end{Lem}

Now we introduce the concept of $q$-global quantum discord ($q$-GQD) as a symmetric generalization of $q$-QD in Lemma~\ref{Def:q-QD2}.
\begin{Def}
For a bipartite state $\rho^{AB}$, its $q$-GQD is defined as
\begin{align}
{\mathcal D}_q\left(\rho^{AB}\right)
=&\min_{\Phi}\left[{\mathcal I}_q\left(\rho^{AB}\right)- {\mathcal I}_q\left(\Phi\left(\rho^{AB}\right)\right)\right],
\label{q-GQD}
\end{align}
where the minimization is over all possible local von Neumann measurements $\Phi=\{\Pi_{j_1}^{A} \otimes \Pi_{j_2}^{B} \}$ on both subsystems $A$ and $B$.
\label{Def:q-GQD}
\end{Def}

We also propose a systematic extension of $q$-GQD in Definition~\ref{Def:q-GQD} into multi-party quantum systems;
as a generalization of quantum mutual information of a $n$-party quantum state $\rho^{A_{1}\cdots A_{n}}$($=\rho^{{\bf A}}$) in Eq.~(\ref{nqmu}),
we define the quantum mutual information in terms of Tsallis-$q$ entropy as
\begin{equation}
{\mathcal I}_q\left(\rho^{{\bf A}}\right) = \sum_{i=1}^{n} S_q\left(\rho^{A_{i }}\right)-S_q\left(\rho^{{\bf A}}\right).
\label{q-nqmu}
\end{equation}

Let us consider a set of local measurements on each subsystem
$\Phi=\{\Pi_{j}= \Pi_{j_{1}}^{A_{1}} \otimes\Pi_{j_{2}}^{A_{2}}\otimes\cdots\otimes\Pi_{j_{n}}^{A_{n}}\}$
and the density operator $\Phi(\rho^{{\bf A}}) = \sum_{j} \Pi_{j}\rho^{{\bf A}}\Pi_{j} $
obtained after the non-selective measurement $\Phi$. The quantum mutual information of $\Phi(\rho^{{\bf A}})$ is then defined as
\begin{align}
{\mathcal  I}_q\left(\Phi\left(\rho^{{\bf A}}\right)\right)=\sum_{i=1}^{n} S_q\left(\Phi\left(\rho^{A_{i }}\right)\right)-S_q\left(\Phi\left(\rho^{{\bf A}}\right)\right)
\label{q-npimu}
\end{align}
with the reduced density matrix
\begin{align}
\Phi\left(\rho^{A_{i }}\right)=&\T_{\overline{A}_i}\Phi\left(\rho^{{\bf A}}\right)
\label{q-npii}
\end{align}
onto subsystem $A_i$ for each $i=1,\cdots, n$.

\begin{Def}
For an $n$-party quantum state $\rho^{A_{1}\cdots A_{n}}$($=\rho^{{\bf A}}$), its $q$-GQD is defined as
\begin{align}
{\mathcal D}_q\left(\rho^{\bf A}\right)
=& \min_{\Phi}\left[{\mathcal  I}_q\left(\rho^{{\bf A}}\right)-{\mathcal  I}_q\left(\Phi\left(\rho^{{\bf A}}\right)\right)\right],
\label{q-nGQD}
\end{align}
where the minimization is over all possible local von Neumann measurements
$\Phi=\{\Pi_{j}= \Pi_{j_{1}}^{A_{1}} \otimes\Pi_{j_{2}}^{A_{2}}\otimes\cdots\otimes\Pi_{j_{n}}^{A_{n}}\}$.
\label{Def:q-nGQD}
\end{Def}
Similar to GQD, $q$-GQD is symmetric under the permutation of subsystems. Moreover,
due to the minimization character over all possible local von Neumann measurements,
it is also clear that $q$-GQD is invariant under local unitary transformations, that is,
\begin{align}
{\mathcal D}_q\left(\rho^{\bf A}\right)=
{\mathcal D}_q\left(U\rho^{\bf A} U^{\dagger}\right),
\label{inva}
\end{align}
for any local unitary operator $U=U^{A_1}\otimes \cdots \otimes U^{A_n}$.

The following theorem shows that the nonnegativity of $q$-GQD is assured for a selective choice of the parameter $q$.
\begin{Thm}
For any $n$-party quantum state $\rho^{\bf A}$, its $q$-GQD is nonnegative for $0< q \leq1$,
\begin{align}
{\mathcal D}_{q}\left(\rho^{\bf A}\right)\geq 0.
\label{q-GQDposi}
\end{align}
\label{Thm:q-GQDposi}
\end{Thm}

\begin{proof}
For a $n$-party quantum state $\rho^{\bf A}$, and a set of local von Neumann measurement
$\Phi=\{\Pi_{j}= \Pi_{j_{1}}^{A_{1}} \otimes\Pi_{j_{2}}^{A_{2}}\otimes\cdots\otimes\Pi_{j_{n}}^{A_{n}}\}$
with the index string $j=(j_{1},~\cdots ,j_{n})$, let $\Phi^{A_k}=\{P_{j_k}^{A_k}=I\otimes \cdots \otimes \Pi _{j_k}^{A_{k}}\otimes \cdots \otimes I\}$ and
\begin{align}
\Phi^{A_k}\left(\rho^{\bf A}\right)=\sum_{j_k}P_{j_k}^{A_k}\rho^{\bf A}P_{j_k}^{A_k},
\label{phik}
\end{align}
for each $k=1,2,\cdots,n$.
In other words, $\Phi^{A_k}\left(\rho^{\bf A}\right)$ is the $n$-party quantum state after the non-selective
local measurement $\{\Pi _{j_k}^{A_{k}} \}$ only
on subsystem $A_k$ for each $k=1,2,\cdots,n$.

We note that for any two subsystems $A_{i}$ and $A_{k}$ such that $1\leq i\neq k \leq n$, we have
\begin{align}
\Phi^{A_i}\left(\Phi^{A_k}\left(\rho^{\bf A}\right)\right)=\Phi^{A_k}\left(\Phi^{A_i}\left(\rho^{\bf A}\right)\right).
\label{ijsame}
\end{align}
Using the notation $\Phi^{A_i}\left(\Phi^{A_k}\left(\rho^{\bf A}\right)\right)=\Phi^{A_i A_k}\left(\rho^{\bf A}\right)$,
we define $\sigma_k^{\bf A}$ as the $n$-party quantum state after the non-selective local measurements
$\{\Pi _{j_1}^{A_{1}}\},~\{\Pi _{j_2}^{A_{2}}\}, \cdots ,\{\Pi _{j_k}^{A_{k}}\}$ on the first $k$ subsystems $A_1, A_2,\cdots,A_k$
for each $k=0,1,2,\cdots,n$, that is,
\begin{align}
\sigma_k^{\bf A}=&\Phi^{A_1 A_2 \cdots A_k}\left(\rho^{\bf A}\right),~\sigma_0^{\bf A}=\rho^{\bf A}
\label{sigk}
\end{align}
and
\begin{align}
\sigma_n^{\bf A}=\Phi^{A_1 A_2 \cdots A_n}\left(\rho^{\bf A}\right)=\Phi\left(\rho^{\bf A}\right)
\label{sig0n}
\end{align}

By assuming that $\Phi=\{\Pi_{j}= \Pi_{j_{1}}^{A_{1}} \otimes\Pi_{j_{2}}^{A_{2}}\otimes\cdots\otimes\Pi_{j_{n}}^{A_{n}}\}$ is an
optimal measurement for ${\mathcal D}_{q}\left(\rho^{\bf A}\right)$, we have
\begin{align}
{\mathcal D}_{q}\left(\rho^{\bf A}\right)=&{\mathcal I}_{q}\left(\rho^{\bf A}\right)-{\mathcal I}_{q}\left(\Phi\left(\rho^{\bf A}\right)\right)\nonumber\\
=&{\mathcal I}_{q}\left(\sigma_0^{\bf A}\right)-{\mathcal I}_{q}\left(\sigma_n^{\bf A}\right)\nonumber\\
=&\sum_{k=0}^{n-1}\left[{\mathcal I}_{q}\left(\sigma_k^{\bf A}\right)-{\mathcal I}_{q}\left(\sigma_{k+1}^{\bf A}\right)\right]\nonumber\\
=&\sum_{k=0}^{n-1}\left[{\mathcal I}_{q}\left(\sigma_k^{\bf A}\right)-{\mathcal I}_{q}\left(\Phi^{A_{k+1}}\left(\sigma_{k}^{\bf A}\right)\right)\right].
\label{recu1}
\end{align}

Now, for each $k=0,\cdots,{n-1}$, let us consider $\sigma_k^{\bf A}=\sigma_k^{A_1 A_2\cdots A_n}$
as a bipartite quantum state $\sigma_k^{\overline{A}_{k+1} A_{k+1}}$ with respect to the bipartition between the subsystems $A_{k+1}$ and its complement,
$\overline{A}_{k+1}$. Then we have
\begin{align}
{\mathcal I}_{q}\left(\sigma_k^{\bf A}\right)&-{\mathcal I}_{q}\left(\Phi^{A_{k+1}}\left(\sigma_{k}^{\bf A}\right)\right)\nonumber\\
\geq&\min_{\Phi^{A_{k+1}}}\left[{\mathcal I}_{q}\left(\sigma_k^{\bf A}\right)-{\mathcal I}_{q}\left(\Phi^{A_{k+1}}\left(\sigma_{k}^{\bf A}\right)\right)\right]\nonumber\\
=&\delta_{q}^{\leftarrow}\left(\sigma_k^{\overline{A}_{k+1} A_{k+1}}\right)\nonumber\\
\geq&0
\label{recu2}
\end{align}
where $\delta_{q}^{\leftarrow}\left(\sigma_k^{\overline{A}_{k+1} A_{k+1}}\right)$ is the $q$-QD of the bipartite quantum state
$\sigma_k^{\overline{A}_{k+1} A_{k+1}}$, and the last inequality holds for $0< q \leq 1$ by Proposition~\ref{Prop:q-QDposi}.
From Eqs.~(\ref{recu1}) and (\ref{recu2}), we have
\begin{align}
{\mathcal D}_{q}\left(\rho^{\bf A}\right)\geq&\sum_{k=0}^{n-1}\delta_{q}^{\leftarrow}\left(\sigma_k^{\overline{A}_{k+1} A_{k+1}}\right)\geq0,
\label{eq:posi2}
\end{align}
which completes the proof.
\end{proof}

\section{Analytic Evaluation}
\label{Sec:analytic}

As a measure of quantum correlation among composite systems, $q$-GQD is a well-defined quantity for arbitrary quantum states.
However, it is hard to evaluate due to the minimization over all possible local von Neumann measurements in the definition.
In this section, by investigating the monotonicity of Tsallis-$q$ entropy under majorization of real vectors,
we provide an analytic way of evaluating $q$-GQD for some classes of multi-qubit quantum states.

Let us recall the definition of Tsallis-$q$ entropy for a probability distribution $P=\{p_j\}$~\cite{tsallis},
\begin{align}
H_q\left(P\right)=-\sum_j {p_j}^q \ln_q p_j=\frac{1-\sum_j {p_j}^q}{q-1}.
\label{ctsal}
\end{align}
For a quantum state $\rho$ with spectral decomposition $\rho=\sum_j \lambda_j\ket{\psi_j}\bra{\psi_j}$,
it is straightforward to check that the quantum Tsallis-$q$ entropy of $\rho$ in Eq.~({\ref{Tqent}}) is
in fact the Tsallis-$q$ entropy of the spectrum of $\rho$, that is,
\begin{align}
S_q\left(\rho\right)=H_q\left(\Lambda\right),
\end{align}
where $\Lambda=\{\lambda_j\}$.

Now we will consider a special property of Tsallis-$q$ entropy, namely {\em Schur concavity}, and
before this, we first introduce the concept of {\em majorization} among real vectors~\cite{MO}.
For real vectors $\vec{x},~ \vec{y} \in \R^n$ such that
$\vec{x}=\left(x_{1}, x_{2}, \cdots, x_{n}\right)$ and  $\vec{y}=\left(y_{1}, y_{2}, \cdots, y_{n}\right)$,
$\vec{x}$ is said to be majorized by $\vec{y}$, denoted by $\vec{x}\prec\vec{y}$ if
\begin{align}
\sum_{j=1}^{k}x_j \leq \sum_{j=1}^{k}y_j
\label{eq:major1}
\end{align}
for $j=1,2,\cdots,n-1$, and
\begin{align}
\sum_{j=1}^{n}x_j = \sum_{j=1}^{n}y_j.
\label{eq:major2}
\end{align}

A real-valued function $\phi$ defined on ${\mathcal A} \subset \R^n$ is said to be
{\em Schur-concave} on ${\mathcal A}$ if
\begin{align}
\vec{x}\prec \vec{y}~~on~~{\mathcal A}~\Rightarrow \phi\left(\vec{x}\right) \geq\phi\left(\vec{y}\right),
\label{eq:Schur}
\end{align}
for any $\vec{x},~ \vec{y} \in {\mathcal A}$.
We further note that $\phi$ is said to be {\em symmetric}
if
\begin{align}
\phi\left(\vec{x} \right)=\phi\left(M\vec{x}\right)
\label{sym}
\end{align}
for any $n$-dimensional permutation $M$, and $\phi$ is said to be {\em concave} if
\begin{align}
\phi\left(p\vec{x}+(1-p)\vec{y}\right)\geq p\phi\left(\vec{x}\right)+(1-p)\phi\left(\vec{y}\right),
\label{phiconcon}
\end{align}
for any $\vec{x},~ \vec{y} \in \R^n$ and $0\leq p \leq1$.

For a sufficient condition of Schur-concavity, we have the following proposition~\cite{MO};
\begin{Prop}
If a real-valued function $\phi$ defined on $\R^n$ is symmetric and concave,
then $\phi$ is Schur-concave.
\label{propchur}
\end{Prop}

Now we show the monotonicity of Tsallis-$q$ entropy under majorization.
\begin{Lem}
For given probability distributions $P=\{p_{1}, p_{2}, \cdots, p_{n} \}$ and $Q=\{q_{1}, q_{2},~ \cdots,~ q_{n} \}$
satisfying $1\geq p_{1} \geq \cdots \geq p_{n} \geq 0$,
$1\geq q_{1} \geq \cdots \geq q_{n} \geq 0$ and $\sum_{i=1}^{n}p_{i}=\sum_{i=1}^{n}q_{i}=1$, if
\begin{align}
\sum_{i=1}^{k}p_{i} \leq  \sum_{i=1}^{k}q_{i},
\label{eq:TsSch1}
\end{align}
for each $k=1, \cdots, n$, then
\begin{align}
H_{q}(P) \geq H_{q} (Q),
\label{eq:TsSch2}
\end{align}
where $H_{q}(P)$ and $H_{q} (Q)$ are Tsallis-$q$ entropies of the probability distributions $P$ and $Q$ respectively.
\label{TsSch}
\end{Lem}

\begin{proof}
For the closed unit interval $I=[0,1]$ of $\R$, we can consider the probability distributions $P$ and $Q$ as
$n$-dimensional vectors in $I^n \subset \R^n$,
\begin{align}
P\leftrightarrow&\vec{p}=\left(p_{1}, p_{2}, \cdots, p_{n}\right),\nonumber\\
Q\leftrightarrow&\vec{q}=\left(q_{1}, q_{2}, \cdots, q_{n}\right).
\label{pvec}
\end{align}
(Thus we will use the notation $H_q\left(\vec{p}\right)$ equivalently as $H_q\left(P\right)$.)

Because $P$ and $Q$ are probability distributions, Eq.~(\ref{eq:TsSch1}) implies that $\vec{p}$ is majorized by
$\vec{q}$. Thus showing Eq.~(\ref{eq:TsSch2}) is equivalent to show that Tsallis-$q$ entropy is Schur-concave.
Furthermore, Proposition~\ref{propchur} says that it is sufficient
to show that Tsalli-$q$ entropy is symmetric and concave.

We first note that Tsallis-$q$ entropy is clearly symmetric by its
definition in Eq.~(\ref{ctsal}), that is, for any probability distribution
$P=\vec{p} \in I^n$,
\begin{align}
S_q\left(\vec{p}\right)=S_q\left(M\vec{p}\right)
\label{tassym}
\end{align}
for any $n$-dimensional permutation $M$.
The concavity of Tsallis-$q$ entropy follows from a simple calculus;
for any probability distribution $\vec{p} \in I^n$
with the normalization condition $\sum_{k=1}^{n} p_k =1$, by introducing the Lagrange multiplier $\lambda$,
we have the Hessian matrix
\begin{align}
\frac{\partial}{\partial p_i \partial p_j}\left(S_q\left(\vec{p}\right)-\lambda\left(\sum_{k=1}^{n}p_k -1\right)\right)=-qp_i^{p-2}\delta_{ij},
\label{Hessian}
\end{align}
which is clearly negative semi-definite for $1\leq i,j\leq n$ and $0<q$. Thus Tsallis-$q$ entropy is Schur-concave for $0<q$, and this completes the proof.
\end{proof}

Let us consider a one-parameter class of multi-qubit states,
\begin{align}
\rho_{\mu} = (1-\mu)\frac{I^{\otimes n}}{2^{n}}+\mu\ket{\psi}\bra{\psi}
\label{wer}
\end{align}
for $\mu \in [0,1]$, where $I$ is the $2\times2$ identity operator and $\ket{\psi}=(\ket{00\cdots 0}+\ket{11\cdots 1})/ \sqrt{2}$.
The class of states in Eq.~(\ref{wer}) is known as the $n$-qubit Werner-GHZ state, which is a mixture of a fully mixed state
(thus no $q$-GQD) and a maximally correlated pure state (thus maximal $q$-GQD).
The following theorem allows an analytic evaluation of $q$-GQD for $n$-qubit Werner-GHZ states.

\begin{Thm}
For the $n$-qubit Werner-GHZ state $\rho_{\mu} = (1-\mu)I^{\otimes n}/2^{n}+\mu\ket{\psi}\bra{\psi}$,
its $q$-GQD is
\begin{align}
{\mathcal D}_q\left(\rho_{\mu}\right)=&\left(\frac{1-\mu}{2^{n}}+\mu\right)^{q}\ln_{q}\left(\frac {1-\mu}{2^{n}} +\mu\right)\nonumber\\
&~+\left(\frac {1-\mu}{2^{n}}\right)^{q}\ln_{q}\left(\frac {1-\mu}{2^{n}}\right) \nonumber\\
&~-2\left(\frac {1-\mu}{2^{n}} +\frac{\mu}{2}\right)^{q}\ln_{q}\left(\frac {1-\mu}{2^{n}}+\frac{\mu}{2}\right).
\label{eq:werfor}
\end{align}
\label{Thm:werfor}
\end{Thm}

An analytic evaluation of GQD for $n$-qubit Werner-GHZ states was proposed in~\cite{Xu}, and our proof method
for Theorem~\ref{Thm:werfor} follows the construction therein.

\begin{proof}
We first note that the reduced density matrix $\rho^{A_k}$ of $\rho_{\mu}$ onto each one-qubit subsystem
$A_1, \cdots A_n$ is proportional to identity operator $I/2$. In this case, it follows from the definition
that the $q$-GQD of $\rho_{\mu}$ is simplified as
\begin{align}
{\mathcal D}_q\left(\rho_{\mu}\right)=-S_q\left(\rho_{\mu}\right)+
\min_{\Phi}S_q\left(\Phi\left(\rho_{\mu}\right)\right)
\label{simple}
\end{align}
with the minimization over all local von Neumann measurements $\Phi$.
Because $\rho_{\mu}$ has $2^n$ eigenvalues
\begin{align}
\{\frac{1-\mu }{2^{n}}+\mu ,\frac{1-\mu }{2^{n}},\frac{1-\mu }{2^{n}},...,\frac{1-\mu }{2^{n}}\},
\label{eigenrho}
\end{align}
$S_q\left(\rho_{\mu}\right)$ in Eq.~(\ref{simple}) can be easily calculated.

We also note that a von Neumann measurement on single qubit can be expressed as
\begin{eqnarray}
\Pi _{0}=\frac{1}{2}\left(I+\vec{\Pi }\cdot \vec{\sigma }\right),   \ \
\Pi _{1}=\frac{1}{2}\left(I-\vec{\Pi }\cdot \vec{\sigma }\right),
\end{eqnarray}
where $\vec{\Pi }=(\alpha ,\beta ,\gamma )$ is a real vector with unit length and
$\vec{\sigma }=\left(\sigma_x, \sigma_y, \sigma_z\right)$ with Pauli matrices
$\sigma_x$, $\sigma_y$ and $\sigma_z$.

Now for any $n$-qubit local projective measurement $\Phi$, we can identify it as
$\Phi=\{\vec{\Pi }_{1}\}\otimes \cdots \otimes \{\vec{\Pi }_{n}\}$ for some $n$ number of real vectors with unit length
$\vec{\Pi }_{i}=(\alpha _{i},\beta_{i},\gamma _{i})$ for $i=1,\cdots,n$.

After a bit of algebra, it can be shown that the quantum state $\Phi\left(\rho_{\mu}\right)$
obtained after a non-selective local von Neumann measurement $\Phi$ on each subsystem has
$2^n$ eigenvalues
\begin{align}
\frac{1-\mu }{2^{n}}+\frac{\mu }{2}[\prod _{i=1}^{n}\frac{1+(-1)^{m_{i}}
\gamma _{i}}{2}+\prod _{i=1}^{n}\frac{I-(-1)^{m_{i}}\gamma _{i}}{2}\nonumber \\
+\prod _{i=1}^{n}\frac{\alpha _{i}+(-1)^{m_{i}}i\beta _{i}}{2}+\prod _{i=1}^{n}%
\frac{\alpha _{i}-(-1)^{m_{i}}i\beta _{i}}{2}],
\label{phirhoeigen}
\end{align}
with $m_{1},m_{2},...,m_{n}\in \{0,1\}.$

Due to the symmetry of Eq.~(\ref{phirhoeigen}) with respect to $m_i$'s, we note that the eigenvalues corresponding to
$\{m_{1},m_{2},...,m_{n}\}$ and $\{1-m_{1},1-m_{2},...,1-m_{n}\}$ are equal. For this reason,
if we consider a real vector consisting $2^n$ eigenvalues of $\Phi\left(\rho_{\mu}\right)$ in Eq.~(\ref{phirhoeigen}),
in decreasing order, it is always majorized by the vector with $\gamma _{i}=1$ for all $i$, that is
\begin{align}
\left(\frac{1-\mu }{2^{n}}+\frac{\mu }{2},\frac{1-\mu }{2^{n}}+\frac{\mu }{2},\frac{1-\mu }{2^{n}},\frac{1-\mu }{2^{n}},...,\frac{1-\mu }{2^{n}}\right).
\label{phirhomon}
\end{align}
Thus by the monotonicity of Tsallis-$q$ entropy under majorization in Lemma~\ref{TsSch}, we assert that
$S_q\left(\Phi\left(\rho_{\mu}\right)\right)$ achieves its minimum when $\gamma _{i}=1$ for all $i$.
From Eqs.~(\ref{eigenrho}) and (\ref{phirhomon}), we are ready to have an analytic evaluation of Eq.~(\ref{simple}),
which leads us to Eq.~(\ref{eq:werfor}).
\end{proof}

Now, let us consider another class of multi-qubit states whose analytic evaluation of $q$-GQD is feasible.

\begin{Thm}
Let us consider an $n$-qubit state
\begin{align}
\rho=\frac{1}{2^{2}}\left(I^{\otimes n}+c_{1}\sigma_{x}^{\otimes n}+c_{2}\sigma _{y}^{\otimes n}+c_{3} \sigma _{z}^{\otimes n}\right)
\label{clas2}
\end{align}
where $I$ is $2 \times 2$ identity operator, and $c_{1}$, $c_{2}$ and $c_{3}$
are real numbers constrained by $0 \leq d=\sqrt{c_1^2+c_2^2+c_3^2}\leq 1$.
If $n$ is odd,
\begin{widetext}
\begin{align}
{\mathcal D}_q(\rho)=-\frac{2^{n-1}}{q-1}\left[\left(\frac {1+c}{2^{n}}\right)^{q}+\left(\frac {1-c}{2^{n}}\right)^{q}
-\left(\frac {1+d}{2^{n}}\right)^{q}-\left(\frac {1-d}{2^{n}}\right)^{q}\right],
\label{forodd}
\end{align}
\end{widetext}
where $c =\max \{|c_{1}|, |c_{2}|, |c_{3}| \}$.
If $n$ is even,
\begin{widetext}
\begin{align}
{\mathcal D}_q(\rho)=-\frac{2^{n-2}}{q-1}\left[2\left(\frac {1+c}{2^{n}}\right)^{q} +2\left(\frac {1-c}{2^{n}}\right)^{q}
-\sum_{j=1}^{4}\left(\frac {\lambda _{j}}{2^{n}}\right)^{q}\right]
\label{foreven}
\end{align}
\end{widetext}
with $\lambda _{1},~\lambda _{2},~\lambda _{3},~\lambda _{4}\in [0,1]$ such that
\begin{align}
\lambda _{1}&= 1+c_{3} +c_{1} +(-1)^{n/2}c_{2},\nonumber\\
\lambda _{2}&= 1+c_{3} -c_{1} -(-1)^{n/2}c_{2},\nonumber\\
\lambda _{3}&= 1-c_{3} +c_{1} -(-1)^{n/2}c_{2},\nonumber\\
\lambda _{4}&= 1-c_{3} -c_{1} +(-1)^{n/2}c_{2}.
\label{lambdas}
\end{align}
\label{The:forclas2}
\end{Thm}

Our proof method for Theorem~\ref{The:forclas2} is based on the construction in~\cite{Xu},
which is an analytic evaluation of GQD for the class of states in Eq.~(\ref{clas2}).

\begin{proof}
Due to the traceless property of Pauli matrices, it is clear that the reduced density matrix $\rho^{A_k}$ of $\rho$
onto each one-qubit subsystem $A_1, \cdots A_n$ is proportional to identity operator $I/2$. Thus the $q$-GQD of $\rho$ is
simplified as
\begin{align}
{\mathcal D}_q\left(\rho\right)=-S_q\left(\rho\right)+
\min_{\Phi}S_q\left(\Phi\left(\rho\right)\right)
\label{simple2}
\end{align}
with the minimization over all local von Neumann measurements $\Phi$.

We also note that $\rho $ has nonzero elements only on the principle diagonal and the antidiagonal, thus a direct calculation
of the characteristic polynomial $det\left(\rho -xI^{\otimes n}\right)=0$ leads us to the eigenvalues of $\rho$;
if $n$ is odd, the
eigenvalues of $\rho $ are
\begin{align}
\{\frac{1}{2^{n}}\left(1\pm \sqrt{c_{1}^{2}+c_{2}^{2}+c_{3}^{2}}\right)\}=\{\frac{1}{2^{n}}(1\pm d)\},
\label{eigodd}
\end{align}
each of them has multiplicity $2^{n-1}$.
When $n$ is even, the eigenvalues of $\rho$ are
\begin{align}
\{\frac{1+c_{3}}{2^{n}}\pm \frac{c_{1}+(-1)^{n/2}c_{2}}{2^{n}},\frac{1-c_{3}%
}{2^{n}}\pm \frac{c_{1}-(-1)^{n/2}c_{2}}{2^{n}}\},
\label{eigeven}
\end{align}
each of them has multiplicity $2^{n-2}$.

Similar to the proof of Theorem~\ref{Thm:werfor}, by identifying an $n$-qubit local projective measurement $\Phi$ with
a set of real vectors $\{\vec{\Pi }_{i}=(\alpha _{i},\beta_{i},\gamma _{i})\}$ for $i=1,\cdots, n$,
the eigenvalues of $\Phi\left(\rho\right)$ can be obtained as
\begin{align}
\{\frac{1}{2^{n}}\left[1\pm \left(c_{1}\prod _{i=1}^{n}\alpha _{i}+c_{2}\prod
_{i=1}^{n}\beta _{i}+c_{3}\prod _{i=1}^{n}\gamma _{i}\right)\right]\},
\label{phieigen2}
\end{align}
each of them has multiplicity $2^{n-1}$.
By Lemma~\ref{TsSch}, we note that minimizing $S_q\left(\Phi\left(\rho\right)\right)$ is equivalent
to maximizing
\begin{align}
\left|c_{1}\prod _{i=1}^{N}\alpha _{i}+c_{2}\prod _{i=1}^{N}\beta _{i}+c_{3}\prod
_{i=1}^{N}\gamma _{i}\right|,
\label{max1}
\end{align}
over all possible vectors $\{\vec{\Pi }_{i}\}_{i=1}^{n}$.

To maximize Eq.~(\ref{max1}), suppose that $\{\vec{\Pi }_{i}\}_{i=1}^{n-1}$ are given, then Eq.~(\ref{max1})
can be regarded as the size of the inner product of two vectors
\begin{align}
\left|\left(\alpha _{n},\beta _{n},\gamma _{n}\right)\cdot \left(c_{1}\prod _{i=1}^{n-1}\alpha
_{i},c_{2}\prod _{i=1}^{n-1}\beta _{i},c_{3}\prod _{i=1}^{n-1}\gamma _{i}\right)\right|.
\label{max2}
\end{align}
Because $\left(\alpha _{n},\beta _{n},\gamma _{n}\right)$ is a unit vector, Eq.~(\ref{max2})
clearly obtains
its maximum
\begin{align}
\left(c_{1}^{2}\prod _{i=1}^{n-1}\alpha _{i}^{2}+c_{2}^{2}\prod _{i=1}^{n-1}\beta
_{i}^{2}+c_{3}^{2}\prod _{i=1}^{n-1}\gamma _{i}^{2}\right)^{1/2},
\label{max3}
\end{align}
that is, $\left(\alpha _{n},\beta _{n},\gamma _{n}\right)$ is in the same direction with
$\left(c_{1}\prod _{i=1}^{n-1}\alpha_{i},c_{2}\prod _{i=1}^{n-1}\beta _{i},c_{3}\prod _{i=1}^{n-1}\gamma _{i}\right)$.

Now suppose $\{\vec{\Pi }_{i}\}_{i=1}^{n-2}$ are given.
Because $\left(\alpha _{n-1},\beta _{n-1},\gamma _{n-1}\right)$ is also a unit vector, a simple calculus implies that
Eq.~(\ref{max3}) obtains its maximum as
\begin{align}
\left(\max\{c_{1}^{2}\prod _{i=1}^{n-2}\alpha _{i}^{2},c_{2}^{2}\prod
_{i=1}^{n-2}\beta _{i}^{2},c_{3}^{2}\prod _{i=1}^{n-2}\gamma _{i}^{2}\}\right)^{1/2}.
\label{max4}
\end{align}
Thus the maximum of Eq.~(\ref{max4}) over all possible $\{\vec{\Pi }_{i}\}_{i=1}^{n-2}$ is clearly
\begin{align}
c=\max\{|c_{1}|,|c_{2}|,|c_{3}|\}.
\label{max5}
\end{align}
In other words, $S_q\left(\Phi\left(\rho\right)\right)$ is Eq.~(\ref{simple2}) can be minimized when
$\Phi\left(\rho\right)$ has eigenvalues
\begin{align}
\{\frac{1}{2^{n}}\left(1\pm c\right)\},
\label{phieigen3}
\end{align}
each of them has multiplicity $2^{n-1}$, that is,
\begin{align}
\min_{\Phi}S_q\left(\Phi\left(\rho\right)\right)=&-2^{n-1}\left[\left(\frac{1+c}{2^n}\right)^q \ln_q \left(\frac{1+c}{2^n}\right)\right]\nonumber\\
&-2^{n-1}\left[\left(\frac{1-c}{2^n}\right)^q \ln_q \left(\frac{1-c}{2^n}\right)\right]\nonumber\\
=&\frac{1-2^{n-1}\left[\left(\frac{1+c}{2^n}\right)^q+\left(\frac{1-c}{2^n}\right)^q\right]}{q-1}.
\label{miniphirho}
\end{align}

If $n$ is odd, Eq.~(\ref{eigodd}) enables us to evaluate $S_q\left(\rho\right)$ in Eq.~(\ref{simple2}), thus together with
Eq.~(\ref{miniphirho}), we have the $q$-GQD of $\rho$ in Eq.~(\ref{forodd}).
For an even $n$, Eq.~(\ref{eigeven}) together with Eq.~(\ref{simple2}) lead us to the $q$-GQD of $\rho$ in Eq.~(\ref{foreven}).
\end{proof}

Due to the continuity of Tsallis entropy with respect to the parameter $q$, Theorems~\ref{Thm:werfor} and \ref{The:forclas2}
recover the results in~\cite{Xu} as a case when $q=1$. We also note that for the case when
$c_{1}=\alpha,~c_{2}=-\alpha, c_{3} =2\alpha-1$ and $n=2$, $\rho$ in Eq.~(\ref{clas2}) is
reduced to a class of two-qubit states, so-called $\alpha$-state, introduced in~\cite{MPP}.
We illustrate the difference of $q$-GQD for two $\alpha$ states (corresponding to $\alpha = 0.58$ and $\alpha = 0.3$)
in Figure~\ref{Fig:GQDdiff}.

\begin{figure}
\centerline{\epsfig{figure=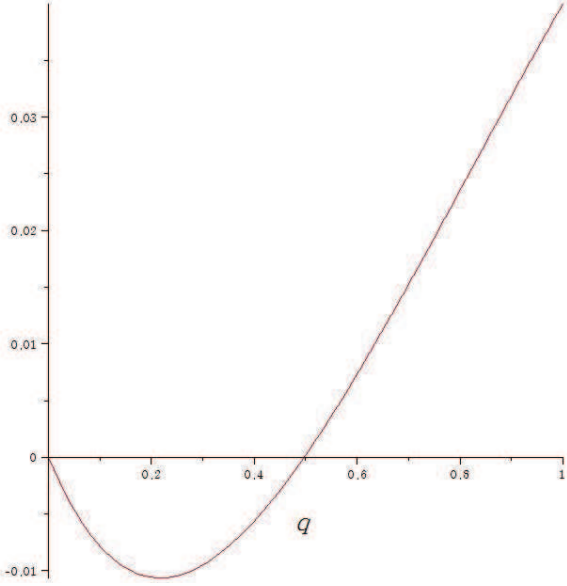, width=0.9\linewidth}}
\caption{\label{Fig:GQDdiff}
Difference between $q$-GQD of two $\alpha$ states for $\alpha = 0.58$ and $\alpha = 0.3$. The vertical axis
represents the difference value of $q$-GQD where the horizontal axis represents the parameter $q$ for $0<q<1$.}
\end{figure}

This difference takes both negative or positive values depending on the range of q. In other words, A relative order among
quantum states in terms of correlation measures is not invariant, because it strongly depends on what entropy function is used to
define the correlation measure.

\section{Monogamy of $q$-GQD in Multi-Party Quantum Systems}
\label{Sec:mono}

In this section, we show that the restricted shareability of $q$-GQD in multi-party quantum system can be
characterized as an inequality, thus $q$-GQD monogamy of multi-party quantum systems.
 By investigating possible decomposition of
multi-party $q$-GQD into bipartite $q$-GQD's among subsystems,
we provide a sufficient condition for a monogamy inequality
of $q$-GQD in multi-party quantum systems.

For a $n$-party quantum state $\rho^{{\bf A}}$, let us consider the difference of
$q$-mutual information in Eq.~(\ref{q-nGQD}) generated by a local von Neumann measurement
$\Phi=\{\Pi_{j}= \Pi_{j_{1}}^{A_{1}} \otimes\Pi_{j_{2}}^{A_{2}}\otimes\cdots\otimes\Pi_{j_{n}}^{A_{n}}\}$,
\begin{align}
{\mathcal D}^{\Phi}_q\left(\rho^{{\bf A}}\right)
=&{\mathcal  I}_q\left(\rho^{{\bf A}}\right)-{\mathcal  I}_q\left(\Phi\left(\rho^{{\bf A}}\right)\right),
\label{ndiff2}
\end{align}
which we will refer to as the $q$-GQD of $\rho^{{\bf A}}$ induced by $\Phi$.

In order to consider a possible decomposition of ${\mathcal D}^{\Phi}_q\left(\rho^{A_{1}\cdots A_{n}}\right)$
in terms of bipartite $q$-GQDs among subsystems,
we further define the $q$-GQD of subsystems concerned with the local von Neumann measurement $\Phi$;
for each $k=1,2,\cdots, n-1$ and the $k+1$-party reduced density matrix $\rho^{A_{1}\cdots A_{k}A_{k+1}}$,
we can consider it as a bipartite quantum state $\rho^{\left(A_{1}\cdots A_{k}\right)A_{k+1}}$
with respect to the bipartition between $A_{1}\cdots A_{k}$ and $A_{k+1}$.
The $q$-mutual information of the bipartite state $\rho^{\left(A_{1}\cdots A_{k}\right)A_{k+1}}$ is then
\begin{align}
{\mathcal I}_q\left(\rho^{\left(A_{1}\cdots A_{k}\right)A_{k+1}}\right)
=&S_q\left(\rho^{A_{1}\cdots A_{k}}\right)+S_q\left(\rho^{A_{k+1}}\right)\nonumber\\
&-S_q\left(\rho^{\left(A_{1}\cdots A_{k}\right)A_{k+1}}\right).
\label{biqmut}
\end{align}
We also have the $q$-mutual information of $\rho^{\left(A_{1}\cdots A_{k}\right)A_{k+1}}$ generated by $\Phi$ as
\begin{widetext}
\begin{align}
{\mathcal I}_q\left(\Phi^{\left(A_{1}\cdots A_{k}\right)A_{k+1}}\left(\rho^{\left(A_{1}\cdots A_{k}\right)A_{k+1}}\right)\right)
=&S_q\left(\Phi^{A_{1}\cdots A_{k}}\left(\rho^{A_{1}\cdots A_{k}}\right)\right)+
S_q\left(\Phi^{A_{k+1}}\left(\rho^{A_{k+1}}\right)\right)\nonumber\\
&-S_q\left(\Phi^{\left(A_{1}\cdots A_{k}\right)A_{k+1}}
\left(\rho^{\left(A_{1}\cdots A_{k}\right)A_{k+1}}\right)\right),
\label{biqmutphi}
\end{align}
\end{widetext}
where $\Phi^{\left(A_{1}\cdots A_{k}\right)A_{k+1}}=\{\left(\Pi_{j_{1}}^{A_{1}}\otimes\cdots\otimes\Pi_{j_{k}}^{A_{k}}\right)\otimes\Pi_{j_{k+1}}^{A_{k+1}}\}$
is the local von Neumann measurement on the first $k+1$ subsystems induced from $\Phi$
and  $\Phi^{\left(A_{1}\cdots A_{k}\right)A_{k+1}}
\left(\rho^{\left(A_{1}\cdots A_{k}\right)A_{k+1}}\right)$ is the density matrix after the non-selective measurement of $\Phi^{\left(A_{1}\cdots A_{k}\right)A_{k+1}}$.
For simplicity, we denote
\begin{align}
\Phi^{\left(A_{1}\cdots A_{k}\right)A_{k+1}}
\left(\rho^{\left(A_{1}\cdots A_{k}\right)A_{k+1}}\right)=\Phi\left(\rho^{\left(A_{1}\cdots A_{k}\right)A_{k+1}}\right),
\label{simple3}
\end{align}
if there is no confusion. The $q$-GQD of the bipartite reduced density matrix $\rho^{\left(A_{1}\cdots A_{k}\right)A_{k+1}}$
generated by the local von Neumann measurement $\Phi$ is then defined as
\begin{align}
{\mathcal D}^{\Phi}_q\left(\rho^{\left(A_{1}\cdots A_{k}\right)A_{k+1}}\right)=&{\mathcal I}_q\left(\rho^{\left(A_{1}\cdots A_{k}\right)A_{k+1}}\right)\nonumber\\
&-{\mathcal I}_q\left(\Phi\left(\rho^{\left(A_{1}\cdots A_{k}\right)A_{k+1}}\right)\right)
\label{inducQDk}
\end{align}
for each $k=1,2,\cdots, n-1$.

Now we have the following theorem about the decomposability of ${\mathcal D}^{\Phi}_q\left(\rho^{A_{1}\cdots A_{n}}\right)$
in terms of bipartite $q$-GQDs among subsystems
\begin{Thm}
For a $n$-party quantum state $\rho^{{\bf A}}=\rho^{A_{1}\cdots A_{n}}$ and
given a non-selective measurement
$\Phi=\{\Pi_{j}= \Pi_{j_{1}}^{A_{1}} \otimes\Pi_{j_{2}}^{A_{2}}\otimes\cdots\otimes\Pi_{j_{n}}^{A_{n}}\}$,
${\mathcal D}^{\Phi}_q\left(\rho^{A_{1}\cdots A_{n}}\right)$
can be decomposed as
\begin{align}
{\mathcal D}^{\Phi}_q\left(\rho^{A_{1}\cdots A_{n}}\right)=\sum_{k=1}^{n-1}
{\mathcal D}^{\Phi}_q\left(\rho^{\left(A_{1}\cdots A_{k}\right)A_{k+1}}\right).
\label{decomp1}
\end{align}
\label{Thm:decomp}
\end{Thm}

\begin{proof}
From the definition of ${\mathcal D}^{\Phi}_q\left(\rho^{A_{1}\cdots A_{n}}\right)$ in Eq.~(\ref{ndiff2})
together with Eqs.~(\ref{q-nqmu}) and (\ref{q-npimu}), we have
\begin{align}
{\mathcal D}^{\Phi}_q\left(\rho^{A_{1}\cdots A_{n}}\right)
=&\sum_{k=1}^{n-1}\left[S_{q}\left(\rho^{A_{k}}\right)
-S_{q}\left(\Phi^{A_{k}}\left(\rho^{A_{k}}\right)\right)\right]\nonumber\\
&+S_{q}\left(\rho^{A_{n}}\right)-S_{q}\left(\Phi^{A_{n}}\left(\rho^{A_{n}}\right)\right)\nonumber\\
&-S_{q}\left(\rho^{A_{1}\cdots A_{n}}\right)+S_{q}\left(\Phi\left(\rho^{A_{1}\cdots A_{n}}\right)\right).
\label{decpf1}
\end{align}

By adding and subtracting $S_{q}\left(\Phi\left(\rho^{A_{1}\cdots A_{n-1}}\right)\right)$ and
$S_{q}\left(\Phi\left(\rho^{A_{1}\cdots A_{n-1}}\right)\right)$ in Eq.~(\ref{decpf1}), we have
\begin{align}
{\mathcal D}^{\Phi}_q\left(\rho^{A_{1}\cdots A_{n}}\right)
={\mathcal D}^{\Phi}_q\left(\rho^{A_{1}\cdots A_{n-1}}\right)+{\mathcal D}^{\Phi}_q\left(\rho^{\left(A_{1}\cdots A_{n-1}\right)A_{n}}\right).
\label{decpf2}
\end{align}
We then iterate this process to decompose ${\mathcal D}^{\Phi}_q\left(\rho^{A_{1}\cdots A_{n-1}}\right)$ in Eq.~(\ref{decpf2}) into
bipartite $q$-GQDs in smaller subsystems, which eventually leads us to Eq.~(\ref{decomp1}).
\end{proof}

In fact, Theorem~\ref{Thm:decomp} reveals a mutually exclusive relation of bipartite $q$-GQDs
shared in multi-party quantum systems; for each term in the summation of the right-hand side of Eq.~(\ref{decomp1}),
we have
\begin{align}
{\mathcal D}^{\Phi}_q\left(\rho^{\left(A_{1}\cdots A_{k}\right)A_{k+1}}\right)\geq
{\mathcal D}_q\left(\rho^{\left(A_{1}\cdots A_{k}\right)A_{k+1}}\right),
\label{phi-GQD}
\end{align}
for any choice of $\Phi$ (and thus $\Phi^{\left(A_{1}\cdots A_{k}\right)A_{k+1}}$) due to the minimization character of $q$-GQD.
If we assume that the local von Neumann measurement
$\Phi=\{\Pi_{j}= \Pi_{j_{1}}^{A_{1}} \otimes\Pi_{j_{2}}^{A_{2}}\otimes\cdots\otimes\Pi_{j_{n}}^{A_{n}}\}$ is optimal for
the $q$-GQD of $\rho^{A_{1}\cdots A_{n}}$, we have
\begin{align}
{\mathcal D}_q\left(\rho^{A_{1}\cdots A_{n}}\right)=&\sum_{k=1}^{n-1}
{\mathcal D}^{\Phi}_q\left(\rho^{\left(A_{1}\cdots A_{k}\right)A_{k+1}}\right)\nonumber\\
\geq&\sum_{k=1}^{n-1}
{\mathcal D}_q\left(\rho^{\left(A_{1}\cdots A_{k}\right)A_{k+1}}\right).
\label{weakmono}
\end{align}
In other words, the summation of bipartite $q$-GQDs shared among subsystems is bounded by the
$q$-GQD of the whole system. The following corollary provides a sufficient condition that the sum of {\em pairwise} quantum
correlations in terms of $q$-GQD is upper limited by the multipartite quantum correlation
with respect to $q$-GQD.

\begin{Cor}
For an $n$-party quantum state $\rho^{A_{1}\cdots A_{n}}$,
\begin{align}
{\mathcal D}_q\left(\rho^{A_{1}\cdots A_{n}}\right) \geq \sum_{k=1}^{n-1}
{\mathcal D}_q\left(\rho^{A_{1}A_{k+1}}\right),
\label{eq:GQDmono}
\end{align}
provided that the $q$-GQD does not increase under discard of subsystems, that is,
\begin{align}
{\mathcal D}_q\left(\rho^{\left(A_{1}\cdots A_{k}\right)A_{k+1}}\right)\geq{\mathcal D}_q\left(\rho^{A_1A_{k+1}}\right)
\label{noinc}
\end{align}
for $k=1,\cdots, n-1$.
\label{Cor:GQDmono}
\end{Cor}

Similar to monogamy inequalities of entanglement~\cite{CKW, OV}, Corollary~\ref{Cor:GQDmono} says that the $q$-GQD of total quantum system
${A_{1}\cdots A_{n}}$ serves as an upper bound for the sum of bipartite $q$-GQD between $A_1$ and each of $A_i$'s for $i=2,\cdots,n$,
provided Inequality~(\ref{noinc}).
We also note that Inequality~(\ref{noinc}) is not necessary for monogamy of $q$-GQD; as remarked in \cite{BROS}, the
three-qubit state $\rho^{ABC}=\left(\ket{000}\bra{000}+\ket{1+1}\bra{1+1}\right)/2$ with $\ket{+}=\left(\ket{0}+\ket{1}\right)/\sqrt 2$ has
vanishing ${\mathcal D}_q\left(\rho^{A(BC)}\right)$, whereas ${\mathcal D}_q\left(\rho^{AB}\right)$ has a nonzero value.
Thus Inequality~(\ref{noinc}) is not generally true. However, it is also straightforward to verify that
${\mathcal D}_q\left(\rho^{AC}\right)=0$ and ${\mathcal D}_q\left(\rho^{ABC}\right)={\mathcal D}_q\left(\rho^{AB}\right)$, thus
the monogamy inequality in (\ref{eq:GQDmono}) is still true for this case.

\section{Conclusion}
\label{Sec:Con}
Using Tsallis-$q$ entropy, we have proposed a class of quantum correlation measures, $q$-GQD,
and showed its nonnegativity for $0<q \leq 1$. We have also provided analytic expressions of
$q$-GQD for multi-qubit Werner-GHZ states and a class of three-parameter multi-qubit states by investigating the monotonicity
of Tsallis-$q$ entropy. We have further provided a sufficient condition for monogamy inequality of $q$-GQD in multi-party quantum systems.

As a one-parameter class of correlation measure, $q$-GQD contains GQD as a special case when $q=1$.
Moreover, we have shown that the relative order among quantum states in terms of correlation measures strongly depends on what entropy
function is used. In other words, the quantum correlations quantified by $q$-GQD
are seen in different manners by distinct entropic quantifiers.

The class of $q$-GQD monogamy inequality provided here also encapsulates the known case of
monogamy inequality in terms GQD as special cases~\cite{BROS}, as well as their explicit
relation with respect to the continuous parameter $q$.
We believe our result provides a useful methodology to understand quantum correlations in multi-party
quantum systems.

\section*{Acknowledgments}
This research was supported by Basic Science Research Program through the National Research Foundation of Korea(NRF)
funded by the Ministry of Education, Science and Technology(2012R1A1A1012246).
D.P.C. was partially supported by Cryptographic Hard Problems Initiatives.

\end{document}